\documentclass{eptcs}

\SetMathAlphabet{\mathcal}{normal}{OMS}{xmdcmsy}{m}{n}
\hypersetup{%
  colorlinks=true,%
  citecolor=blue,%
  urlcolor=blue,%
  linkcolor=blue,%
}

\providecommand{\event}{IWS 2010} %
\usepackage{breakurl}             %
\usepackage{tikz}
\usetikzlibrary{calc,snakes}
\usepackage{ifthen}
\usepackage{amssymb}
\usepackage{wrapfig}
\usepackage{graphicx} 
\usepackage{amsmath}  
\usepackage{array} 
\usepackage{arrows} 
\usepackage{float}
\usepackage[english]{babel}
\usepackage{cite}
\usepackage{amsthm}

\newtheorem{theorem}{Theorem}
\newtheorem{corollary}[theorem]{Corollary}

\newtheorem{lemma}[theorem]{Lemma}
\newtheorem{definition}[theorem]{Definition}
\newtheorem{example}[theorem]{Example}

\renewcommand{\theenumi}{\roman{enumi}}

\makeatletter
\renewcommand{\p@enumii}{\theenumi--}
\makeatother

\newcommand{\ul}[1]{\underline{#1}}
\newcommand{\Fo}{\mathsf{1}}
\newcommand{\BBox}{{\Box\hspace{0.1em}}}

\newcommand{\mpair}[2]{#1 \gtrdot #2} %
\newcommand{\sMp}[2]{\textcolor{black}{(\mpair{#1}{#2}, \mu)}} %
\newcommand{\eMp}[4]{\textcolor{black}{(\mpair{#1}{#2},#3,#4,\mu)}} %

\newcommand{\MM}{\mathcal{M}}

\newcommand{\rDF}[1]{Definition~\ref{#1}}
\newcommand{\rCO}[1]{Corollary~\ref{#1}}
\newcommand{\rEX}[1]{Example~\ref{#1}}
\newcommand{\rLES}[2]{\rLE{#1} \rSub{#2}}
\newcommand{\rLE}[1]{Lemma~\ref{#1}}
\newcommand{\lFI}[1]{\label{Figure:#1}}
\newcommand{\rFI}[1]{\text{Figure~\ref{Figure:#1}}}

\newcommand{\rTH}[1]{Theorem~\ref{#1}}
\newcommand{\rSC}[1]{Section~\ref{#1}}
\newcommand{\rSub}[1]{{(\ref{#1})}}

\newcommand{\RR}{\mathcal{R}}
\newcommand{\SSS}{\mathcal{S}}
\newcommand{\QQ}{\mathcal{Q}}

\newcommand{\VV}{\mathcal{V}}
\newcommand{\WW}{\mathcal{W}}

\newcommand{\PP}{\mathcal{P}}
\newcommand{\Ff}{\mathsf{f}}
\newcommand{\Fsnd}{\mathsf{2nd}}

\newcommand{\Fg}{\mathsf{g}}
\newcommand{\Fh}{\mathsf{h}}
\newcommand{\Fa}{\mathsf{a}}

\newcommand{\Pos}{\mathcal{P}\mathsf{os}}

\newcommand{\Fs}{\mathsf{s}}

\newcommand{\Fz}{\mathsf{0}}

\renewcommand{\neq}{\mathrel{/\kern-1.9ex=}}

\newcommand{\monus}[1][\displaystyle]{\mathrel{\smash{\!\!\!\!\!\begin{array}[c]{c}
#1\cdot \\[-2.2ex]
#1-\\[0.7ex]
\end{array}\!\!\!\!\!}}}
\newcommand{\oto}[1][]%
{\mathrel{\smash{\stackrel{\raisebox{2pt}{\scriptsize $\mathsf{o}\:$}}%
{\smash{\rightarrow}}}_{#1}}}
\newcommand{\loto}[1][]%
{\mathrel{\smash{\stackrel{\raisebox{2pt}{\scriptsize $\mathsf{lo}\:$}}%
{\smash{\rightarrow}}}_{#1}}}
\newcommand{\lto}[1][]%
{\mathrel{\smash{\stackrel{\raisebox{2pt}{\scriptsize $\mathsf{l}\:$}}%
{\smash{\rightarrow}}}_{#1}}}
\newcommand{\poto}[1][]%
{\mathrel{\smash{\stackrel{\raisebox{2pt}{\scriptsize $\mathsf{po}\:$}}%
{\smash{\rightarrow}}}_{#1}}}
\newcommand{\pito}[1][]%
{\mathrel{\smash{\stackrel{\raisebox{2pt}{\scriptsize $\mathsf{pi}\:$}}%
{\smash{\rightarrow}}}_{#1}}}
\newcommand{\moto}[1][]%
{\mathrel{\smash{\stackrel{\raisebox{2pt}{\scriptsize $\makebox[2mm]{$\mathsf{mo}$}\:$}}%
{\smash{\rightarrow}}}_{#1}}}
\newcommand{\mito}[1][]%
{\mathrel{\smash{\stackrel{\raisebox{2pt}{\scriptsize $\mathsf{mi}\:$}}%
{\smash{\rightarrow}}}_{#1}}}
\newcommand{\pto}[1][]%
{\mathrel{\smash{\stackrel{\raisebox{2pt}{\scriptsize $\mathsf{p}\:$}}%
{\smash{\rightarrow}}}_{#1}}}
\newcommand{\mto}[1][]%
{\mathrel{\smash{\stackrel{\raisebox{2pt}{\scriptsize $\mathsf{m}\:$}}%
{\smash{\rightarrow}}}_{#1}}}
\newcommand{\ito}[1][]%
{\mathrel{\smash{\stackrel{\raisebox{2pt}{\scriptsize $\mathsf{i}\:$}}%
{\smash{\rightarrow}}}_{#1}}}
\newcommand{\lito}[1][]%
{\mathrel{\smash{\stackrel{\raisebox{2pt}{\scriptsize $\mathsf{li}\:$}}%
{\smash{\rightarrow}}}_{#1}}}
\newcommand{\itos}[2]%
{\mathrel{\smash{\stackrel{\raisebox{2pt}{\scriptsize $\mathsf{i}\:$}}%
{\smash{\rightarrow}}}_{#1}^{#2}}}
\newcommand\nats{\mathbb{N}}

\newcommand{\rto}[1]{\to_{#1}}
\newcommand{\Pito}[1][]%
{\mathrel{\smash{\stackrel{\raisebox{2pt}{\scriptsize $\Pi\:$}}%
{\smash{\rightarrow}}}_{#1}}}
\newcommand{\muto}[1][]%
{\mathrel{\smash{\stackrel{\raisebox{2pt}{\scriptsize $\mu\:$}}%
{\smash{\rightarrow}}}_{#1}}}

\newcommand{\qto}[1][]%
{\mathrel{\smash{\stackrel{\raisebox{3pt}{\scriptsize $\QQ\:$}}%
{\smash{\rightarrow}}}_{#1}}}
\newcommand{\tor}[1][]{\to_{#1}}

\renewcommand{\epsilon}{\varepsilon}

\newcommand{\Finf}{\mathsf{inf}}
\newcommand{\Ffact}{\mathsf{factorial}}
\newcommand{\Ffa}{\mathsf{fact}}

\newcommand{\Fif}{\mathsf{if}}

\newcommand{\Ftrue}{\mathsf{true}}
\newcommand{\Fnonloop}{\mathsf{nloop}}
\newcommand{\Ffalse}{\mathsf{false}}
\newcommand{\FFeq}{==}
\newcommand{\FFti}{\cdot}
\newcommand{\FFpl}{+}
\newcommand{\Fpl}[2]{#1 \FFpl #2}
\newcommand{\Feq}[2]{#1 \mathrel{\FFeq} #2}

\newcommand{\Fe}[2]{\mathsf{eq}(#1,#2)}

\newcommand{\Fcheck}{\mathsf{chk}}

\newcommand{\Fti}[2]{#1 \mathrel{\FFti} #2}

\newcommand{\br}[1]{(#1)}
\newcommand{\arraytag}[1][]{& \stepcounter{equation} (\arabic{equation})\ifthenelse{\equal{#1}{}}{}{\newcounter{#1}\setcounter{#1}{\value{equation}}}}
\newcommand{\reftag}[1]{(\arabic{#1})}

\newcommand\TTTT{%
 \ensuremath{\textsf{T\kern-0.2em\raisebox{-0.3em}T\kern-0.2emT\kern-0.2em\raisebox{-0.3em}2}}%
}

\newcommand{\eat}[1]{}

\newcommand{\drawcontexton}[3]{%
\draw (#1) -- ++(-2,-1) -- ++(4,0) -- node[above right] {$C$} (#1);
\node[shape = rectangle, transform shape, draw=black,fill=white, yshift=-1cm,xshift=0.9cm] (#2) at (#1) {};
\draw[snake=snake, line before snake=2mm,segment length=1cm] (#1) to node[left] {$p$} (#2);
\coordinate[yshift=-1cm,xshift=-1.3cm] (#3) at (#1);
\coordinate[yshift=-5mm, xshift = -5mm] (#1zwi) at (#1);
\problemD[iii]{#1zwi};
}

\newcommand{\drawterml}[3]{%
\draw (#1) -- node[right] {$t$} ++(1,-1.4) -- ++(-2,0) -- (#1);
\coordinate[yshift=-7mm,xshift=-3mm] (#1wi) at (#1);
\coordinate[yshift=-7mm,xshift=+1.5mm] (#1zwit) at (#1);
\coordinate[yshift=-1.4cm,xshift=0.45cm] (#2) at (#1);
\coordinate[yshift=-1.4cm,xshift=-0.75cm] (#3) at (#1);
\draw[snake=snake, line before snake=2mm, segment length=0.6cm] (#1) to node[left] {$q$} (#1zwit);
\redex{#1zwit};
\problemD[i]{#1wi}
\draw (#1zwit) -- ++(-0.5,-0.7) -- +(1,0) -- (#1zwit);
}

\newcommand{\drawmub}[2]{%
\draw (#1) -- node[right] {$\mu$} ++(0.5,-0.6) -- ++(-1,0) -- (#1);
\coordinate[yshift=-0.6cm] (#2) at (#1);
}

\newcommand{\drawmup}[3][]{%
\draw (#2) -- node[right] {$\mu$} ++(0.5,-0.6) -- ++(-1,0) -- (#2);
\coordinate[yshift=-0.4cm,xshift=-0.5mm] (#2zwi) at (#2);
\problemD[#1]{#2zwi} circle (2pt);
\coordinate[yshift=-0.6cm] (#3) at (#2);
}

\tikzstyle{problem}=[fill=white,draw=black]
\tikzstyle{redex}=[fill=black]

\newcommand{\problemD}[2][B]{
\fill[problem] ($(#2) + (-1mm,0)$) circle (2pt);
\fill[redex] ($(#2) + (-1mm,0)$) circle (0.5pt);
\node[node distance=2.0mm,right of=#2] {\tiny(\ref{#1})};
}
\newcommand{\redex}[2]{
\fill[redex] (#1) circle (2pt);
}

\newcommand{\drawcontext}[3]{%
\draw (#1) -- ++(-2,-1) -- ++(4,0) -- (#1);
\node[yshift=-0.6cm] at (#1) {\ensuremath{C}};
\node[shape = rectangle, transform shape, draw=black,fill=white, yshift=-1cm,xshift=+0.9cm] (#2) at (#1) {};
\coordinate[yshift=-1cm,xshift=-1.3cm] (#3) at (#1);
}

\newcommand{\drawterm}[2]{%
\draw (#1) -- ++(1,-1) -- ++(-2,0) -- (#1);
\node[yshift=-0.6cm] at (#1) {\ensuremath{t}};
\coordinate[yshift=-1cm,xshift=-0.3cm] (#2) at (#1);
}

\newcommand{\drawmu}[2]{%
\fill (#1) circle (2pt);
\draw (#1) -- ++(0.5,-0.6) -- ++(-1,0) -- (#1);
\node[yshift=-0.35cm] at (#1) {\ensuremath{\mu}};
\coordinate[yshift=-0.6cm] (#2) at (#1);
}

\title{Loops under Strategies \dots{} Continued}
\def\titlerunning{Loops under Strategies \dots{} Continued}
\author{%
Ren\'e Thiemann\thanks{%
  These authors are supported by the FWF (Austrian Science Fund) project P22767-N13.
}
\quad\qquad
Christian Sternagel\setcounter{footnote}{0}\footnotemark{}
\institute{University of Innsbruck\\ Austria}
\email{\{rene.thiemann, christian.sternagel\}@uibk.ac.at}
\and
J\"urgen Giesl\thanks{%
 This author is supported by the DFG (German Research Foundation)
 project GI 274/5-3.
}
\institute{RWTH Aachen University\\ Germany}
\email{giesl@informatik.rwth-aachen.de}
\and
Peter Schneider-Kamp
\institute{University of Southern Denmark\\ Denmark}
\email{petersk@imada.sdu.dk}
}
\def\authorrunning{Ren\'e Thiemann et.\ al}

\begin{document}

\maketitle

\begin{abstract}
While there are many approaches for automatically \emph{proving}
termination of term rewrite systems, up to now there exist only few
techniques to \emph{disprove} their termination automatically. Almost all
of these techniques try to find \emph{loops}, where the existence of a loop
implies non-termination of the rewrite system.  However, most programming
languages use specific \emph{evaluation strategies}, whereas loop detection
techniques usually do not take strategies into account. So even if a
rewrite system has a loop, it may still be terminating under certain
strategies.

Therefore, our goal is to develop decision procedures which can determine
whether a given loop is also a loop under the respective evaluation
strategy. In earlier work, such procedures were presented for the strategies of
innermost, outermost, and context-sensitive evaluation. In the current
paper, we build upon this work and develop such decision procedures
for important strategies like leftmost-innermost, leftmost-outermost,
(max-)parallel-innermost, (max-)parallel-outermost, and forbidden patterns
(which generalize  innermost, outermost, and context-sensitive strategies).
In this way, we obtain the first  approach to disprove termination under
these strategies automatically.%
\end{abstract}

\section{Introduction}
\label{s:intro}

Termination is an important property of term rewrite systems (TRSs).
Therefore, much effort has been spent on developing and automating
techniques for showing termination of TRSs.  However, in order to detect
bugs, it is at least as important to prove \emph{non-termination}.
Note
that for rewriting under a strategy, the strategy has to be taken into account
when
checking for non-termination. The reason is that a TRS which is
non-terminating when ignoring the strategy may still be
terminating when considering the strategy. Thus, it is important to develop automated
techniques to disprove termination of TRSs under strategies. 

Most of the techniques for showing non-termination detect \emph{loops} (for
example, 
\cite{Frocos05,GKM83,Kurth,LM78,Pay08,MATCHBOX,TORPA_JAR}). For a TRS $\RR$,
a loop is a derivation of the form \mbox{$t \to^+_\RR C[t\mu]$} for some
context $C$ and some substitution $\mu$.  To prove non-termination
under a strategy $\SSS$, we may use a complete
transformation $T_\SSS$ (e.g., \cite{GM04,Trafo,SOFSEM09}) where a TRS $\RR$
terminates under the strategy $\SSS$ iff the TRS $T_\SSS(\RR)$ terminates
when ignoring the strategy. After applying such a transformation, we may
try to find a loop in the transformed system $T_\SSS(\RR)$.  However, there
are some drawbacks: The first problem is an increased search space, as
loops of $\RR$ are often transformed into much longer loops in
$T_\SSS(\RR)$.  Moreover, the complete transformations from
\cite{GM04,Trafo,SOFSEM09} translate a loop $t \to_\RR^+ C[t\mu]$ into a
non-looping infinite derivation in $T_\SSS(\RR)$, whenever $C \neq \BBox$.
These two problems were solved in \cite{RTA08,RTA09} by decision
procedures which, given a loop in the original system $\RR$, directly
decide whether the loop is also a loop under the respective strategy.  Here,
\cite{RTA08} treats the innermost strategy whereas \cite{RTA09} deals with
the context-sensitive \cite{Luc00} and the outermost strategy.
Another problem is the availability of complete transformations.  For
the leftmost-innermost, parallel-innermost, and max-parallel-innermost
strategy we know by \cite{Kri00} that a TRS is terminating under one of
these strategies iff it is innermost terminating. Thus, we can use the
decision procedure for innermost loops \cite{RTA08} to disprove
termination under these strategies.\footnote{By performing all steps in
a parallel-innermost step one after another, one can easily show
that innermost loopingness is equivalent to parallel-innermost loopingness.
Moreover, by \cite{Kri00} an innermost loop  implies
leftmost-innermost and max-parallel innermost non-termination. Yet, this does neither imply
leftmost-innermost nor max-parallel-innermost loopingness. 
As an example, consider $\RR' = \{\Fa \to
\Ff(\Fnonloop,\Fa)\} \cup \RR$, where $\Fnonloop$ is a non-terminating, but
non-looping term w.r.t.\ $\RR$. Then $\RR'$ is innermost looping but neither
leftmost-innermost looping nor max-parallel-innermost looping. 
This might be a motivation to also develop decision
procedures for the various innermost strategies. Since the decision procedures
for leftmost-, parallel-, and max-parallel-\emph{outermost} loops directly
also give us decision procedures for the corresponding \emph{innermost}
strategies, we will mention these results in the paper as well.} However, we are not aware
of any complete transformation for the strategies leftmost-outermost,
parallel-outermost, and max-parallel-outermost.  Therefore, in this paper
we build upon the direct methods of \cite{RTA08,RTA09} and give decision
procedures for all these strategies (i.e., these procedures again decide
whether a loop is also a loop under the strategy).  Note that our decision procedures can
also be extended to the context-sensitive case, e.g., to the leftmost-innermost
context-sensitive strategy.

Finally, recently a generalization of
innermost\,/\,outermost\,/\,context-sensitive rewriting has been introduced:
\emph{rewriting with forbidden patterns} 
\cite{GS09}. In this paper we also
develop a decision procedure for loops under forbidden
patterns.

Before giving an overview on the contents of this paper, we present a
motivating example.
\begin{example}
\label{ex1}
Consider the following TRS (computing the factorial) which 
is a variant of \cite[Ex.\ 1]{RTA08}.\vspace{-1.5ex}
\begin{minipage}{0.55\textwidth}
\[
\begin{array}{r@{\ }lr}
\Ffact(y) & \to \Ffa(\Fz,y) \arraytag\\
\Ffa(x,y) & \to \Fif(\Feq{x}y,\Fs(\Fz),\Fti{\Ffa(\Fs(x),y)}{\Fs(x)}) \arraytag[fact]\\
\Fif(\Ftrue,x,y) & \to x \arraytag\\
\Fif(\Ffalse,x,y) & \to y \arraytag\\
\Fpl{\Fz}y & \to y \arraytag[plStart]\\
\Fpl{\Fs(x)}y & \to \Fs(\Fpl{x}y) \arraytag
\end{array}
\]
\end{minipage}
\begin{minipage}{0.45\textwidth}
\[
\begin{array}{r@{\ }lr}
\Fti{\Fz}y & \to \Fz \arraytag\\
\Fti{\Fs(x)}y & \to \Fpl{y}{\br{\Fti{x}y}} \arraytag\\
 {\Feq{x}y} & \to {\Fe{\Fcheck(x)}{\Fcheck(y)}} \arraytag[eqStart]\\
 {\Fe{x}x} & \to \Ftrue \arraytag[eRule]\\
{\Fcheck(x)} & \to \Ffalse \arraytag\\
 {\Fe{\Ffalse}y} & \to \Ffalse \arraytag[eqEnd]
\end{array}
\]
\end{minipage}

\medskip\noindent
Here, $\Ffa(x,y)$ computes $\prod_{x \leq z < y} (z+1) = (x+1) \cdot (x+2) \cdot \ldots \cdot y$. The intended strategy is leftmost-outermost. Otherwise,  rule
\reftag{fact} would directly cause non-ter\-mi\-na\-tion. Moreover, this strategy
is needed for the equality-test encoded by rules
\reftag{eqStart}--\reftag{eqEnd} (which takes at most three reductions). 
Nevertheless, we obtain the following looping leftmost-outermost reduction
(the respective redexes are underlined):
\begin{align*}
t & = \ul{\Ffa(x,y)}
\\ & \to \Fif(\ul{\Feq{x}y},\Fs(\Fz),\Fti{\Ffa(\Fs(x),y)}{\Fs(x)}) 
\\ & \to \Fif(\Fe{\ul{\Fcheck(x)}}{\Fcheck(y)},\Fs(\Fz),\Fti{\Ffa(\Fs(x),y)}{\Fs(x)}) 
\\ & \to \Fif(\ul{\Fe\Ffalse{\Fcheck(y)}},\Fs(\Fz),\Fti{\Ffa(\Fs(x),y)}{\Fs(x)}) 
\\ & \to \ul{\Fif(\Ffalse,\Fs(\Fz),\Fti{\Ffa(\Fs(x),y)}{\Fs(x)})}
\\ & \to \Fti{\Ffa(\Fs(x),y)}{\Fs(x)} 
\\& = C[t\mu]
\end{align*} 
where $\mu = \{x / \Fs(x) \}$ and $C = \Fti\BBox{\Fs(x)}$.  Applying our
new decision procedure developed in this paper
will show that the above loop indeed is a leftmost-outermost
loop, and hence, $\RR$ does not terminate under the leftmost-outermost
strategy.
\end{example}

The rest of the paper is structured as follows: In \rSC{s:strategy loops}
we give the necessary preliminaries.
Afterwards, in \rSC{s:leftmost loops}, we treat the special case of
leftmost loops. Next, in \rSC{s:maxpar loops}, we consider parallel as well as
max-parallel loops. Subsequently,
we handle the more complicated
case of loops under forbidden patterns in \rSC{s:forbidden loops}.
 Finally, in \rSC{s:concl}, we conclude.

\section{Preliminaries}
\label{s:strategy loops}
We only regard finite signatures  and TRSs
and refer to \cite{BN98} for the basics of rewriting. We use $\ell$, $r$,
$s$, $t$, $u$ for terms, $f$, $g$ for function symbols, $x$, $y$ for variables,
$\mu$, $\sigma$ for substitutions, $i$, $j$, $k$, $n$, $m$ for natural
numbers, $o$, $p$, $q$ for positions, 
and $C$, $D$ for
contexts. Here, contexts are terms which contain exactly one hole $\BBox$. 
A position $p$ is left of $q$ iff $p = o \, i \, p'$, $q = o \, j \, q'$, and $i < j$.
The set of variables is denoted by $\VV$.

Throughout this paper we assume a fixed TRS $\RR$ and we write $t \tor[p]
s$ if one can reduce $t$ to $s$ at position $p$ with $\RR$, i.e., $t =
C[\ell\sigma]$ and $s = C[r\sigma]$ for some rule $\ell \to r \in \RR$,
substitution $\sigma$, and context $C$ with $C|_p = \BBox$. In this case,
the term $\ell\sigma$ is called a redex at position $p$.  The reduction is
leftmost\,/\,innermost\,/\,outermost, written $t
\lto[p]$\;/\;$\ito[p]$\;/\;$\oto[p] s$, iff $p$ is a
leftmost\,/\,innermost\,/\,outermost position of $t$ where $t|_p$ is
a redex.  The leftmost-innermost reduction is defined as ${\lito[p]}
= {\lto[p]} \cap {\ito[p]}$.  Similarly, the leftmost-outermost reduction
is ${\loto[p]} = {\lto[p]} \cap {\oto[p]}$.  If the position is irrelevant
we just write $\tor$, $\lto$, $\ito$, $\oto$, $\lito$, and $\loto$,
respectively.

We also consider parallel reductions. Here, $t \pto[q_1,\dots,q_k] s$ is a parallel reduction iff
$k > 0$, the $q_i$'s are pairwise parallel positions, and $t \rto{q_1} \dots \rto{q_k} s$.
The max-parallel reduction relation is defined by $t \mto[q_1,\dots,q_k] s$
iff $t \pto[q_1,\dots,q_k] s$ 
and $t$ has no further redex 
at a position that is parallel to all positions $q_1,\dots,q_k$. 
The (max-)parallel-innermost reduction is defined by $t \mito / \pito[q_1,\dots,q_k] s$ iff $t \mto / \pto[q_1,\dots,q_k] s$
and all redexes $t|_{q_i}$ are innermost redexes. The (max-)parallel-outermost reductions $\moto$ and $\poto$ are defined
analogously.

To shortly illustrate the difference between the strategies, observe that
for the TRS $\RR$ of \rEX{ex1}, $\Feq{x}{y} \ito^* / \lito^* / \mito^* /
\oto^* / \moto^* \Ftrue$ whereas $\Feq{x}{y} \mathrel{/\kern-1.9ex\loto^*}
\Ftrue$. Moreover, $\Feq{\Fz}{\Fz} \ito^* / \lito^* / \mito^* \Ffalse$ but
$\Feq\Fz\Fz \oto^* / \loto^* / \moto^* \Ffalse$ is not possible.

Next, we consider rewriting under \emph{forbidden patterns}.

\begin{definition}[Rewriting under forbidden patterns \cite{GS09}]
A \emph{forbidden pattern} is a triple $(\ell,o,\lambda)$ for a term
$\ell$, position \mbox{$o \in \Pos(\ell)$}, and $\lambda \in \{h,a,b\}$.
For a set $\Pi$ of forbidden patterns the induced rewrite relation $\Pito$
is defined by $t \Pito_p s$ iff $t \to_p s$ and there is no pattern
$(\ell,o,\lambda) \in \Pi$ such that there exist a position $o' \in
\Pos(t)$, a substitution $\sigma$ with $t|_{o'} = \ell\sigma$, and
\begin{itemize}
\item $p = o'o$, if $\lambda = h$,
\item $p < o'o$, if $\lambda = a$, and
\item $p > o'o$, if $\lambda = b$.
\end{itemize}
\end{definition}

So a forbidden pattern $(\ell,o,h)$ means that the redex may not be at
position $o$ in a subterm of the form $\ell \sigma$. Similarly,
$(\ell,o,a)$ and $(\ell,o,b)$ mean that the redex may not be strictly
above and not strictly below position $o$ in a subterm of the form $\ell
\sigma$, respectively. 

Several strategies are expressible using $\Pito$ \cite{GS09}:
Innermost rewriting is obtained by setting $\Pi = \{(\ell,\epsilon,a) \mid
\ell \to r \in \RR\}$, outermost rewriting by using $\Pi =
\{(\ell,\epsilon,b) \mid \ell \to r \in \RR\}$, $\QQ$-restricted-rewriting
\cite{LPAR04} by $\Pi = \{(\ell,\epsilon,a) \mid \ell \to r \in
\QQ\}$, and context-sensitive-rewriting \cite{Luc00} w.r.t.\ the
replacement map $\mu$ can be expressed by $\Pi =
\{(f(x_1,\ldots,x_n),i,\lambda) \mid f \in \Sigma, i \notin
\mu(f), \lambda \in \{h,b\}\}$, where $\Sigma$ is the set of all function
symbols of the signature.

However, even more sophisticated examples can be treated by forbidden patterns.
\begin{example}
\label{2nd}
Consider the following TRS from \cite{GS09,Luc01}.
\begin{align*}
\Finf(x) & \to x : \Finf(\Fs(x)) \\
\Fsnd(x : (y : zs)) & \to y
\end{align*}
This TRS is not weakly normalizing, but still some terms like $\Fsnd(\Finf(\Fz))$ have a normal form.
One purpose of forbidden patterns is to restrict the rewrite relation in such a way that the restriction
is terminating, but that all normal forms are still being reached. Here, context-sensitive
rewriting is too restrictive, since forbidding rewriting in the second
argument of ``$:$''
would not
allow the reduction $\Fsnd(\Finf(\Fz)) \to \Fsnd(\Fz : \Finf(\Fs(\Fz))) \to \Fsnd( \Fz : (\Fs(\Fz) : \Finf(\Fs(\Fs(\Fz)))))
\to \Fs(\Fz)$. However, we can use rewriting with forbidden patterns where
$\Pi$ only contains the pattern $(x : (y : \Finf(z)), 2.2, h)$. Note that
$(x : (y : \Finf(z)))|_{2.2} = \Finf(z)$. 
Then, $\Pito$ is terminating, but the above reduction is still allowed.
\end{example}

A TRS $\RR$ is \mbox{non-terminating} iff there is an infinite derivation
$t_1 \tor t_2 \tor \cdots$. It is leftmost-innermost / leftmost-outermost /
 parallel-innermost / parallel-outermost /
max-parallel-innermost / max-parallel-outermost / forbidden pattern  non-terminating iff there is
such an infinite derivation using $\lito$ / $\loto$ / $\pito$ /
$\poto$ / $\mito$ / $\moto$  / $\Pito$ instead of $\tor$.  To describe the infinite
derivation that is induced by a loop, we use context-substitutions.

\begin{definition}[Context-substitutions \cite{RTA09}]
A \emph{context-substitution} is a pair $(C,\mu)$ consisting of
a context $C$ and a substitution $\mu$. The $n$-fold application of 
$(C,\mu)$ to a term $t$, written $t(C,\mu)^n$, is defined as follows.
\begin{xalignat*}{3} 
t(C,\mu)^{0} & = t & t(C,\mu)^{n+1} &= C[t(C,\mu)^n\mu]
\end{xalignat*}
\end{definition}
\begin{wrapfigure}[8]%
{r}{5.65cm}
\vspace*{-.5cm}
\begin{tikzpicture}[scale=0.8]
\begin{scope}[transform shape]
\coordinate (start)  at (0,0);
\drawcontext{start}{h0}{ma0}
\drawcontext{h0}{h1}{m0}
\drawcontext{h1}{h2}{m1}
\drawterm{h2}{m2}
\drawmu{m1}{m5}
\drawmu{m2}{m3}
\drawmu{m3}{m4}
\drawmu{m0}{m00}
\drawmu{m4}{m40}
\drawmu{m5}{m50}
\node at  (-3,0) {}; 
\end{scope}
\end{tikzpicture}
\caption{\lFI{illustration}The term $t(C,\mu)^3$}
\end{wrapfigure}

\noindent
For example, $t(C,\mu) = C[t\mu]$, 
$t(C,\mu)^2 = C[C[t\mu]\mu] = C[C\mu[t\mu^2]]$, etc. So in general,
in $t(C,\mu)^n$, the context~$C$ is added $n$-times above
$t$ and $t$ is instantiated by $\mu^n$. Note that also the added contexts
are instantiated by $\mu$.
For the term~$t(C,\mu)^3$ this is 
illustrated in \rFI{illustration}.
Context-substitutions have similar
properties to contexts and substitutions.
\begin{lemma}[Properties of context-substitutions \cite{RTA09}]
\label{properties}\ 
\begin{enumerate}
\item\label{prop1} $t(C,\mu)^{n}\mu = t\mu(C\mu,\mu)^{n}$.
\item\label{prop2} $t(C,\mu)^{m}(C,\mu)^{n} = t(C,\mu)^{m+n}$.
\item\label{prop3} If $C|_p = \BBox$ then $t(C,\mu)^{n}|_{p^n} = t\mu^n$.
\item\label{prop4} Whenever $t \to_{q} s$ and $C|_{p} = \BBox$ then
  $t(C,\mu)^{n} \to_{p^{n}q} s(C,\mu)^{n}$.
\end{enumerate}
\end{lemma}
Here, property \rSub{prop1} is similar to the fact that $C[t]\mu =
C\mu[t\mu]$, and \rSub{prop2} shows that context-substitutions can be
combined just like substitutions where $\mu^m\mu^n = \mu^{m+n}$.
Property (iii) shows that the $n$-fold application of $(C,\mu)$ to $t$ yields
a term containing the $n$-fold application of $\mu$ to $t$.
Finally, stability and monotonicity of
rewriting are used to show in \rSub{prop4} that rewriting is closed under
context-substitutions. 
Using context-substitutions we can now concisely 
present the infinite derivation resulting from a loop $t \tor^{+} C[t\mu] = t(C,\mu)$.
\begin{equation*}
\label{looping reduction}
t(C,\mu)^{0}
 \tor^{+} t(C,\mu)^0(C,\mu) = t(C,\mu)^{1}
 \tor^{+} 
 \cdots
 \tor^{+} t(C,\mu)^n
 \tor^{+}
 \cdots
\end{equation*}
So for every $n$, the
positions of the reductions in the loop are
prefixed by an 
additional $p^{n}$ where $p$ is the position of the hole in $C$, 
cf.~\rLES{properties}{prop4}.

\begin{definition}[$\SSS$-loops \cite{RTA09}]
\label{Sloop}
Let $\SSS$ be a strategy.\footnote{In this paper we use a rather liberal definition of a strategy:
a strategy is just a restriction of the rewrite relation.} %
A loop $t_1 \to_{q_1} t_2 \to_{q_2} \cdots \to_{q_m} t_{m+1} = t_1(C,\mu)$ 
with $C|_p = \BBox$ is an \emph{$\SSS$-loop} iff
the reduction $t_i(C,\mu)^n \to_{p^nq_i} t_{i+1}(C,\mu)^n$ respects the
strategy $\SSS$ for all $i \leq m$ and all $n \in \mathbb{N}$.
\end{definition}

As a direct consequence of \rDF{Sloop}, we can conclude that every $\SSS$-loop
of a rewrite system $\RR$ proves non-termination of $\RR$ under the strategy
$\SSS$. Moreover, \rDF{Sloop} also shows that being a loop is a modular property 
in the following sense.
\begin{corollary}[Loops of intersection strategies]
\label{loops intersection}
Let $\SSS$, $\SSS_1$, and $\SSS_2$ be strategies such that 
${\stackrel{\SSS}\to_p} = {\stackrel{\SSS_1}\to_p} \cap {\stackrel{\SSS_2}\to_p}$
for all positions $p$.
Then a loop is an $\SSS$-loop iff it is both an $\SSS_1$-loop and an $\SSS_2$-loop.
\end{corollary}
Hence, to decide whether a loop is leftmost-innermost / leftmost-outermost, 
we just require a decision procedure for leftmost loops and a decision procedure
for innermost / outermost loops. As decision procedures for innermost loops and 
outermost loops have already been developed \cite{RTA08,RTA09}, it remains to
construct a decision procedure for leftmost loops (see \rSC{s:leftmost
loops}).

For rewriting with forbidden patterns, we observe that ${\Pito_p} = 
\bigcap_{(\ell,o,\lambda) \in \Pi} {\singlearrow{\{(\ell,o,\lambda)\}}_p}$,
and hence, by \rCO{loops intersection} it suffices to consider loops w.r.t.\
single forbidden patterns which is the content of \rSC{s:forbidden loops}.

\section{Leftmost Loops}
\label{s:leftmost loops}

\begin{wrapfigure}[14]{r}{9cm}
\vspace{-8mm}
\hfill
\begin{tikzpicture}
\begin{scope}[transform shape]
\coordinate (start)  at (0,0);
\drawcontexton{start}{h0}{ma0}
\drawcontexton{h0}{h1}{m0}
\drawcontexton{h1}{h2}{m1}
\drawterml{h2}{m2}{m6}
\drawmup[iv]{m1}{m5}
\drawmub{m2}{m3}
\drawmub{m3}{m4}
\drawmup[ii]{m6}{m7}
\drawmup[ii]{m7}{m8}
\drawmup[ii]{m8}{m9}
\drawmup[iv]{m0}{m00}
\drawmub{m4}{m40}
\drawmup[iv]{m5}{m50}
\end{scope}
\end{tikzpicture}
\hfill\,
\caption{\lFI{leftmost}Leftmost redexes}
\end{wrapfigure}

Recall the definition of $\lto$. A leftmost reduction of all terms
$t(C,\mu)^n$ at positions $p^nq$ requires that for no $n$ there is a redex
at a position left of $p^nq$.  This is illustrated in \rFI{leftmost}: The
reduction of the subterm at the black position $p^nq$ respects the leftmost
strategy iff $p^nq$ is leftmost. This is the case whenever there are no
redexes at positions $\odot$.

We want to be able to decide whether all $p^nq$ point to
leftmost redexes in the term $t(C,\mu)^n$.  There are four possibilities why $p^nq$ might not point
to a leftmost redex in that term. These cases are marked with (i)-(iv) in  \rFI{leftmost}.
\begin{enumerate}
  \item \label{i}
  There might be a redex within $t\mu^n$ at a position $q' \in \Pos(t)$
  which is left of $q$. Hence, we have to consider all finitely many
  subterms $u = t|_{q'}$ where $q'$ is left of $q$ and guarantee that
  $u\mu^n$ is no redex. 

  \item \label{ii}
  There might be a redex within $t\mu^n$ at a position $q' \in \Pos(t\mu^n)
  \setminus \Pos(t)$ which is left of $q$.  Hence, this redex is of the
  form $u\mu^k$ for some $k \leq n$ and some subterm $u \unlhd x\mu$  where
  $x$ is a variable that occurs within some of $v$, $v\mu$, $v\mu^2$,
  $\dots$ for some subterm $v = t|_{q'}$ where $q'$ is left of $q$.\footnote{
  It does not suffice to only consider the variables $x$ that occur in $v$
  and $v\mu$. This can be seen for $v = y$ and $\mu = \{ y / y_1, y_1 / y_2, y_2 / y_3, 
  \dots y_{n-1} / x, x / f(\dots)\}$. Here, $x$ does neither occur in $v$ nor
  in $v\mu$, but in $v\mu^n$.
  Hence,  the potential redex $f(\dots)$ is detected 
  only after $n$ iterations.}
  Note that there are only finitely many such variables $x$ and hence, again we
  obtain a finite set of terms where for each of these terms $u$ and each
  $n$ we have to guarantee that $u\mu^n$ is not a redex.	

  \item \label{iii}
  There might be a redex where the root is within $C$ and left of the path
  $p$.  Here, we have to consider all finitely many subterms $u = C|_{p'}$
  where $p'$ is left of $p$ and guarantee that $u\mu^n$ is not a redex. 

  \item \label{iv}
  In analogy to \rSub{ii} we also have to consider redexes within $\mu$
  where now the variables $x$ are taken from the subterms $u = C|_{p'}$ where
$p'$ is left of $p$.
\end{enumerate}

\noindent
To summarize, we generate a finite set $U$ of terms $u$ such that  (a) and (b)
 are equivalent:
\begin{itemize}
\item[(a)]
    For every $n$, the reduction $t(C,\mu)^n \to_{p^nq} t'(C,\mu)^n$ is
leftmost.
\item[(b)]
    There is no $u \in U$ and no number $n$ such that $u\mu^n$ is a redex.
\end{itemize}
Note that the question whether $u\mu^n$ is a redex for some $n$ can be
formulated as the kind of matching problem that was encountered for
deciding innermost loops.
\begin{definition}[Matching problems \cite{RTA08}]
A \emph{matching problem} is a pair 
$\sMp u\ell$. 
It is \emph{solvable} iff there are $n$ and $\sigma$ such that $u\mu^n = \ell\sigma$.
\end{definition}
Thus, following the possibilities (i) - (iv) above, we can formally define
a set of matching problems to analyze leftmost reductions.  
\begin{definition}[Leftmost matching problems]
\label{leftmost problems}
The set of \emph{leftmost matching problems} for a reduction $t \to_q t'$ and a context-substitution $(C,\mu)$ with $C|_p = \BBox$ 
is defined as the set consisting of:
\begin{align*}
\text{$\sMp{u}\ell$} & \text{ for each $\ell \to r \in \RR$ and $q' \in \Pos(t)$ where $q'$ is left of $q$, and $u = t|_{q'}$}  \\
\text{$\sMp{u}\ell$} & \text{ for each $\ell \to r \in \RR$ and $q' \in \Pos(t)$ where $q'$ is left of $q$,
 $x \in \bigcup_{i \in \nats} \VV(t|_{q'}\mu^i)$, and $u \unlhd x\mu$} \\
\text{$\sMp{u}\ell$} & \text{ for each $\ell \to r \in \RR$ and $p' \in \Pos(C)$ where $p'$ is left of $p$, and $u = C|_{p'}$}  \\
\text{$\sMp{u}\ell$} & \text{ for each $\ell \to r \in \RR$ and $p' \in \Pos(C)$ where $p'$ is left of $p$,
 $x \in \bigcup_{i \in \nats} \VV(C|_{p'}\mu^i)$, and $u \unlhd x\mu$}
\end{align*}
\end{definition}
Note that the sets of variables in the second and fourth case are finite and
can easily be computed.
The above considerations prove the following theorem.
\begin{theorem}[Soundness of leftmost matching problems]
\label{leftmost MPs correct}
Let $t \to_q t'$ and let $(C,\mu)$ be a context-sub\-sti\-tu\-tion such that $C|_p = \BBox$. 
All reductions $t(C,\mu)^n \to_{p^nq} t'(C,\mu)^n$ are leftmost 
iff none of the leftmost matching problems for $t \to_q t'$ and $(C,\mu)$
is solvable.
\end{theorem}
Using \rTH{leftmost MPs correct} in combination with the decision procedures for matching
problems yields the following corollary.
\begin{corollary}[Leftmost loops are decidable]
\label{leftmost loops}
Let there be a loop $t_1 \to_{q_1} t_2 \to_{q_2} \cdots \to_{q_m} t_{m+1} = t_1(C,\mu)$ 
with $C|_p = \BBox$. Then it is decidable whether the loop is a leftmost loop. 
\end{corollary}
Combining \rCO{leftmost loops} and \rCO{loops intersection} with the decision procedures 
for innermost and outermost loops of \cite{RTA08,RTA09} yields decision procedures
which determine whether a given loop is a leftmost-innermost loop or a leftmost-outermost loop:
for each loop construct the leftmost matching problems, ensure that all these matching
problems are not satisfiable (then leftmost reductions are guaranteed), and
moreover use the 
decision procedures of \cite{RTA08,RTA09} to further ensure that the loop is an innermost 
or outermost loop.

\begin{corollary}[Leftmost-innermost and leftmost-outermost loops are decidable]
\label{corl}
Let there be a loop $t_1 \to_{q_1} t_2 \to_{q_2} \cdots \to_{q_m} t_{m+1} = t_1(C,\mu)$ 
with $C|_p = \BBox$. Then the following two questions are decidable.
\begin{itemize}
\item Is the loop a leftmost-innermost loop?
\item Is the loop a leftmost-outermost loop? 
\end{itemize}
\end{corollary}

\begin{example}
Using \rCO{corl}, we can decide that the loop given in \rEX{ex1} is a
leftmost loop, since for this loop, the set of leftmost matching problems
is empty (as there is never a position left of the used redex). 
Moreover, by the results of \cite{RTA08,RTA09} we can decide
that the loop is an outermost loop, but not an innermost loop.  Hence, the
loop is a leftmost-outermost loop, but not a leftmost-innermost loop.
\end{example}

\begin{example}
\label{innermost loop}
We consider the following loop for the TRS of \rEX{ex1}
\begin{align*}
t & = \ul{\Ffa(x,y)}
\\ & \to \Fif(\ul{\Feq{x}y},\Fs(\Fz),\Fti{\Ffa(\Fs(x),y)}{\Fs(x)}) 
\\ & \to \Fif(\Fe{\ul{\Fcheck(x)}}{\Fcheck(y)},\Fs(\Fz),\Fti{\Ffa(\Fs(x),y)}{\Fs(x)}) 
\\ & \to \Fif(\Fe\Ffalse{\ul{\Fcheck(y)}},\Fs(\Fz),\Fti{\Ffa(\Fs(x),y)}{\Fs(x)}) 
\\ & \to \Fif(\ul{\Fe\Ffalse\Ffalse},\Fs(\Fz),\Fti{\Ffa(\Fs(x),y)}{\Fs(x)}) 
\\ & \to \Fif(\Ffalse,\Fs(\Fz),\Fti{\Ffa(\Fs(x),y)}{\Fs(x)}) 
\\& = C[t\mu]
\end{align*} 
  where $C = \Fif(\Ffalse,\Fs(\Fz),\Fti\BBox{\Fs(x)})$ and $\mu = \{x / \Fs(x)\}$.
  We decide that this loop is a leftmost loop by constructing
  the leftmost matching problems
  \begin{itemize}
  \item $\sMp \Ffalse\ell$ for all left-hand sides $\ell$ (due to the reduction $\Fif(\Fe\Ffalse{\ul{\Fcheck(y)}},\dots) \to \dots$)
  \item $\sMp \Ffalse\ell$, $\sMp \Fz\ell$, and $\sMp{\Fs(\Fz)}\ell$ for all left-hand sides $\ell$ (since $C = \Fif(\Ffalse,\Fs(\Fz),\Fti\BBox{\dots})$)
  \end{itemize}
  and observing that none of them is solvable. This loop is also an innermost loop, but not an 
  outermost loop and hence, it is a leftmost-innermost loop, but not a leftmost-outermost loop.
\end{example}

Whereas in the previous two examples it is rather easy to see that the loops are leftmost, since the leftmost matching
problems are trivially not solvable, we now 
present two more examples where the resulting matching problems are more involved.

\begin{example}
Consider the TRS
\begin{align*}
\Ff(x,y,z) & \to \Fh(\Fg(x,y), \Ff(y,z,z)) \\
\Fg(x,x) & \to x
\end{align*}
and the loop $t = \Ff(x,y,z) \to \Fh(\Fg(x,y), \Ff(y,z,z)) = C[t\mu]$ for $C = \Fh(\Fg(x,y),\BBox)$ and $\mu = \{x/y,y/z\}$.
Here, we construct the non-solvable leftmost matching problems $\sMp{u}\ell$ for all left-hand sides $\ell$ and $u \in \{x,y,z\}$. But additionally
we construct the leftmost matching problem $\sMp{\Fg(x,y)}{\Fg(x,x)}$ which is solvable, since $\Fg(x,y)\mu^2 = \Fg(y,z)\mu = \Fg(z,z)
= \Fg(x,x)\sigma$ for $\sigma = \{x/z\}$. Hence, the loop is not a leftmost loop.
\end{example}

\begin{example}
Consider the TRS
\begin{align*}
\Ff(x,y,z) & \to \Fh(\Fg(x), \Ff(y,z,\Fs(x))) \\
\Fg(\Fs(\Fs(\Fs(x)))) & \to x
\end{align*}
and the loop $t = \Ff(x,y,z) \to \Fh(\Fg(x), \Ff(y,z,\Fs(x))) = C[t\mu]$ for $C = \Fh(\Fg(x),\BBox)$ and $\mu = \{x/y,y/z,z/\Fs(x)\}$.
Here, we construct the non-solvable leftmost matching problems $\sMp{u}\ell$ for all left-hand sides $\ell$ and $u \in \{x,y,z,\Fs(x)\}$. But additionally
we construct the leftmost matching problem $\sMp{\Fg(x)}{\Fg(\Fs(\Fs(\Fs(x))))}$ which is solvable, 
since $\Fg(x)\mu^9 = \Fg(\Fs(\Fs(\Fs(x))))$. Hence, the loop is not a leftmost loop.
\end{example}

\section{Parallel and Max-Parallel Loops}
\label{s:maxpar loops}

For
the parallel innermost\,/\,outermost strategies it suffices to use the decision 
procedures for innermost- and outermost loops. The reason is that 
$t(C,\mu)^n \pto[p^nq_1,\dots,p^nq_k] t'(C,\mu)^n$ is a $\pito$\,/\,$\poto$-reduction
iff for every $1 \leq i \leq k$ there is some $s_i$ such that $t(C,\mu)^n \rto{p^nq_i} s_i$ 
is an innermost\,/\,outermost reduction. 

Hence, for the rest of the section we consider the max-parallel strategies $\mito$ and
$\moto$. Again, the innermost or outermost aspect can be decided by the 
respective decision procedures using a variant of \rCO{loops intersection}
where one allows parallel rewrite steps. 
It remains to consider the max-parallel 
aspect, i.e., we have to decide whether 
$t(C,\mu)^n \mto[p^nq_1,\dots,p^nq_k] t'(C,\mu)^n$ for all $n$.

Here, we essentially proceed as in the leftmost case, where we replace the condition
that some position is left of $p$ or $q$ by the condition that it is parallel to $p$
or to each $q_i$.

\begin{definition}[Max-parallel matching problems]
\label{parallel problems}
The set of \emph{max-parallel matching problems} for a reduction 
$t \pto[q_1,\dots,q_k] t'$ and a context-substitution $(C,\mu)$ with $C|_p = \BBox$ 
is defined as the set consisting of:
\begin{align*}
\text{$\sMp{u}\ell$} & \text{ for each $\ell \to r \in \RR$ and $q' \in
\Pos(t)$ where $q'$ is parallel to all positions $q_i$, and $u = t|_{q'}$}  \\
\text{$\sMp{u}\ell$} & \text{ for each $\ell \to r \in \RR$ and $q' \in \Pos(t)$  where $q'$ is parallel to all $q_i$,
 $x \in \bigcup_{i \in \nats} \VV(t|_{q'}\mu^i)$, and $u \unlhd x\mu$} \\
\text{$\sMp{u}\ell$} & \text{ for each $\ell \to r \in \RR$ and $p' \in
\Pos(C)$ where $p'$ is parallel to $p$, and $u = C|_{p'}$}  \\
\text{$\sMp{u}\ell$} & \text{ for each $\ell \to r \in \RR$ and $p' \in \Pos(C)$ where  $p'$ is parallel to $p$,
 $x \in \bigcup_{i \in \nats} \VV(C|_{p'}\mu^i)$, and $u \unlhd x\mu$}
\end{align*}
\end{definition}
Using this finite set of matching problems we again obtain a decision procedure.
\begin{theorem}[Soundness of max-parallel matching problems]
\label{parallel MPs correct}
Let $t \pto[q_1,\dots,q_k] t'$ and let $(C,\mu)$ be a context-substitution
such that $C|_p = \BBox$. 
All reductions $t(C,\mu)^n \pto[p^nq_1,\dots,p^nq_k] t'(C,\mu)^n$ are max-parallel
iff none of the max-parallel matching problems for
$t \pto[q_1,\dots,q_k] t'$ and $(C,\mu)$
 is solvable.
\end{theorem}

\begin{corollary}[Max-parallel loops are decidable]
Let $t_1 \pto[q^1_1,\dots,q^1_{k_1}] t_2 \pto[q^2_1,\dots,q^2_{k_2}] \cdots
\pto[q^m_1\dots q^m_{k_m}] t_{m+1}$ be a loop with $t_{m+1} = t_1(C, \mu)$
and $C|_p = \BBox$. Then the following questions are decidable.
\begin{itemize}
\item Is the loop a max-parallel loop?
\item Is the loop a parallel-innermost loop? Is it a max-parallel-innermost loop?
\item Is the loop a parallel-outermost loop? Is it a max-parallel-outermost loop?
\end{itemize}
\end{corollary}

Note that in the corollary we did not list the question ``Is the loop a parallel loop?'' since
every loop is trivially also a parallel loop.

\begin{example}
It is easy to see that neither the loop of \rEX{ex1} nor the loop of
\rEX{innermost loop} is a max-parallel loop. The reason is that both loops
violate the max-parallel strategy already in the second reduction step.
However, the TRS of \rEX{ex1} is both max-parallel-outermost and -innermost looping which
is proved by the following two loops 
which could be obtained automatically using a loop detection technique 
and our decision procedure of \rTH{parallel MPs correct}.

The max-parallel-outermost loop needs two parallel reductions:

\vspace*{-.3cm}

{\footnotesize
\begin{align*}
t & = \Fif(\ul{\Fe\Ffalse\Ffalse},\Fo,\Fti{\Fif(\Fe{\ul{\Fcheck(\Fs(x))}}{\ul{\Fcheck(y)}},\Fo,
	\Fti{\Fif(\ul{\Feq{\Fs^2(x)}y},\Fo,\Fti{\ul{\Ffa(\Fs^3(x),y)}}{\Fs^3(x)})}{\Fs^2(x)})}{\Fs(x)})
\\  & \moto \ul{\Fif(\Ffalse,\Fo,\Fti{\Fif(\Fe\Ffalse\Ffalse,\Fo,
	\Fti{\Fif(\Fe{\Fcheck(\Fs^2(x))}{\Fcheck(y)},\Fo,\Fti{
	 \Fif(\Feq{\Fs^3(x)}y,\Fo,\Fti{\Ffa(\Fs^4(x),y)}{\Fs^4(x)})
	}{\Fs^3(x)})}{\Fs^2(x)})}{\Fs(x)})}
\\  & \moto \Fti{\Fif(\Fe\Ffalse\Ffalse,\Fo,
	\Fti{\Fif(\Fe{\Fcheck(\Fs^2(x))}{\Fcheck(y)},\Fo,\Fti{
	 \Fif(\Feq{\Fs^3(x)}y,\Fo,\Fti{\Ffa(\Fs^4(x),y)}{\Fs^4(x)})
	}{\Fs^3(x)})}{\Fs^2(x)})}{\Fs(x)}
\\& = C[t\mu]
\end{align*} 
}

\vspace*{-.3cm}

\noindent
where $C = \Fti\BBox{\Fs(x)}$, $\mu = \{x / \Fs(x)\}$, and where $\Fo$
abbreviates $\Fs(\Fz)$.  For the max-parallel-innermost loop one parallel
reduction suffices:
  {\footnotesize
\begin{align*}
t & = \Fif(\ul{\Fe\Ffalse\Ffalse},\Fo,\Fti{\Fif(\Fe{\ul{\Fcheck(\Fs(x))}}{\ul{\Fcheck(y)}},\Fo,
	\Fti{\Fif(\ul{\Feq{\Fs^2(x)}y},\Fo,\Fti{\ul{\Ffa(\Fs^3(x),y)}}{\Fs^3(x)})}{\Fs^2(x)})}{\Fs(x)})
\\  & \mito \Fif(\Ffalse,\Fo,\Fti{\Fif(\Fe\Ffalse\Ffalse,\Fo,
	\Fti{\Fif(\Fe{\Fcheck(\Fs^2(x))}{\Fcheck(y)},\Fo,\Fti{
	 \Fif(\Feq{\Fs^3(x)}y,\Fo,\Fti{\Ffa(\Fs^4(x),y)}{\Fs^4(x)})
	}{\Fs^3(x)})}{\Fs^2(x)})}{\Fs(x)})
\\& = C[t\mu]
\end{align*} 
}
  where $C = \Fif(\Ffalse,\Fo,\Fti\BBox{\Fs(x)})$ and $\mu = \{x / \Fs(x)\}$.
\end{example}

\section{Loops for Rewriting with Forbidden Patterns}
\label{s:forbidden loops}
For rewriting with forbidden patterns we have to investigate for given
$t$, $t'$, $C$, $\mu$ with $C|_p = \BBox$ and $t \to_q t'$, whether all reductions
$t(C,\mu)^n \to_{p^nq} t'(C,\mu)^n$ are allowed w.r.t.~some fixed forbidden pattern $(\ell,o,\lambda)$.
In other words, we have to check whether 
\begin{equation}
\label{fb}
\text{there are $n$, $o'$, and $\sigma$ with $t(C,\mu)^n|_{o'} = \ell\sigma$ and }
\begin{cases}
 p^nq = o'o, & \text{ if $\lambda = h$,} \\
 p^nq < o'o, & \text{ if $\lambda = a$, and} \\
 p^nq > o'o, & \text{ if $\lambda = b$.}
\end{cases}
\end{equation}
In the subsections \ref{subsection h}-\ref{subsection b}, we investigate the three cases of $\lambda$. We
show that for all of them, \rSub{fb} is decidable.  To this end, we reuse
algorithms that have been developed to decide innermost and outermost
loops.

\subsection{Deciding Loops for Forbidden Patterns of Type $(\cdot,\cdot,h)$}
\label{subsection h}
We start with the easiest case where $\lambda = h$.  Given $p$, $q$, and
$o$, here we first want to
figure out whether there are $n$ and $o'$ such that the condition $p^nq =
o'o$ of \rSub{fb} is satisfied. Obviously, once $n$ has been fixed, then
$o'$ is uniquely determined. Therefore, we first compute $n_0$ as the minimal 
value of $n$ such that $p^nq = o'o$ is
satisfied for some $o'$ and then uniquely determine $o_0'$ such 
that $p^{n_0}q = o_0'o$.

This can be done as follows. If $p = \epsilon$, then one can set $n_0 = 0$
and just has to determine whether $q$ has $o$ as a suffix. Otherwise, one
has to ensure that $p^nq$ is at least as long as $o$. This is done by
choosing $n_0 = \lceil\frac{|o| \monus[\scriptstyle] |q|}{|p|}\rceil$.
If there is an $n$ where  $\exists o'. p^nq = o'o$ can be satisfied,
then $n_0$ is the minimal such number.
Here, ``$\monus$'' is the subtraction on natural numbers where $x \monus y = \max(x-y,0)$.
Afterwards one
just checks whether $p^{n_0}q$ contains $o$ as suffix. If this holds, then
there is
obviously a unique $o_0'$ such that $p^{n_0}q = o_0'o$.
Otherwise, there cannot be any $n$ and $o'$ which satisfy $p^{n}q = o'o$.
The reason is that for any solution $p^{n}q = o'o$ we know that $n \geq
n_0$ and hence, $p^{n - n_0}p^{n_0}q = p^{n}q = o'o$ shows that $o$ is a
suffix of $p^{n_0}q$ as $|p^{n_0}q| \geq |o|$.

In this way we can compute the minimal number $n_0$ and the
corresponding $o_0'$ such that $p^{n_0}q = o_0'o$, or we detect that $p^nq
= o'o$ is unsatisfiable.  In the latter case we are finished since we know
that the forbidden pattern will not restrict any of the desired reductions.
In the former case we can represent the set of solutions  of $p^nq = o'o$
conveniently:
\[
\{(n,o') \mid p^nq = o'o\} = \{(k+n_0,p^k o_0') \mid k \in \nats \}
\]
Hence, it remains to check whether there are $k \in \nats$ and 
$\sigma$ with $t(C,\mu)^{k+n_0}|_{p^k o_0'} = \ell\sigma$. Note that this problem can be simplified
using \rLE{properties}:
\[
t(C,\mu)^{k+n_0}|_{p^k o_0'} 
= t(C,\mu)^{n_0}(C,\mu)^k|_{p^k}|_{o_0'}
= t(C,\mu)^{n_0}\mu^k|_{o_0'}
= (t(C,\mu)^{n_0}|_{o_0'}) \mu^k
\]
Thus,  for the concrete terms $u =
t(C,\mu)^{n_0}|_{o_0'}$ and $\ell$, we have to decide whether there are $k$ and $\sigma$ such that
$u\mu^k = \ell\sigma$. 

\begin{definition}[$(\ell,o,h)$ matching problems]
The set of \emph{$(\ell,o,h)$ matching problems} for a term $t$, a position
$q \in \Pos(t)$,
and a context-substitution $(C,\mu)$ with $C|_p = \BBox$ 
is defined as 
\begin{itemize}
\item the empty set, if there are no $n$ and $o'$ such that $p^nq = o'o$
\item $\{\sMp{t(C,\mu)^{n_0}|_{o_0'}}\ell\}$, otherwise, 
  where $n_0$ and $o_0'$ form the unique minimal solution to the
  equation $p^nq = o'o$
\end{itemize}
\end{definition}
By the discussion above, we have proved the following theorem. 
\begin{theorem}[Soundness of $(\ell,o,h)$ problems]
\label{h MPs correct}
Let $t \to_q t'$ and let $(C,\mu)$ be a context-substitution such that $C|_p = \BBox$. 
All reductions $t(C,\mu)^n \to_{p^nq} t'(C,\mu)^n$ are allowed
w.r.t.\ the pattern $(\ell,o,h)$ iff none of the $(\ell,o,h)$ matching 
problems for $t$, $q$,  and $(C,\mu)$ is solvable.
\end{theorem}
Using \rTH{h MPs correct} in combination with the decision procedure of \cite{RTA08} 
for solvability of matching problems, one can decide whether all reductions 
$t(C,\mu)^n \to_{p^nq} t'(C,\mu)^n$ are allowed
w.r.t.\ the pattern $(\ell,o,h)$.

\begin{example}
We consider the TRS of \rEX{2nd}
and $\Pi = \{(x : (y : \Finf(z)), 2.2, h)\}$. Here, we have the looping reduction
$t = \Finf(x) \to x : \Finf(\Fs(x)) = C[t\mu]$ for $C = x : \BBox$ and
$\mu = \{x / \Fs(x)\}$. Hence, to investigate whether this loop is a
$\Pi$-loop, we have
$p = 2$ as the position of $\,\BBox$ in $C$,
$q = \epsilon$ since the reduction is on the root position of $t$,
and $o = 2.2$. Then we compute $n_0 = \lceil\frac{|o| \monus[\scriptstyle] |q|}{|p|}\rceil = 
\lceil\frac{2 \monus[\scriptstyle] 0}{1}\rceil = 2$ and observe that $p^{n_0}q = 2.2$ has $o = 2.2$ as a suffix, and set 
$o_0' = \epsilon$. Hence, we construct the matching problem $\sMp{t(C,\mu)^{n_0}|_{o_0'}}\ell
= \sMp{\Finf(x)(C,\mu)^2}\ell
= \sMp{x : (\Fs(x) : \Finf(\Fs(\Fs(x))))}{x : (y : \Finf(z))}$
which is solvable because $(x : (\Fs(x) : \Finf(\Fs(\Fs(x)))))\mu^n = 
(x : (y : \Finf(z)))\sigma$
by choosing $n = 0$ and $\sigma = \{y/\Fs(x), z / \Fs(\Fs(x))\}$.
Thus, by \rTH{h MPs correct} we know that this loop is not a $\Pi$-loop.
\end{example}
\subsection{Deciding Loops for Forbidden Patterns of Type $(\cdot,\cdot,a)$}
\label{subsection a}
Also for patterns of type $(\cdot,\cdot,a)$ we want to generate a finite set of matching
problems such that the loop respects a pattern $(\ell,o,a)$ iff none of these 
matching problems is solvable. 
Essentially, we replace the condition $p^nq = o'o$ of the previous subsection
by $p^nq < o'o$, i.e., $o'o$ must now be strictly below the redex. 

The plan is to systematically represent all terms $t(C,\mu)^n|_{o'}$
for all numbers $n$ and all positions $o'$ where $p^nq < o'o$.
We
consider two alternatives: either the term starts within $C^n[t]$ and not
in the substitutions below $t$, or the term starts within the substitutions
that are below $t$. To distinguish these possibilities, we define the finite
set of positions $\PP = \{q' \mid qq' \in \Pos(t)\}$.  Then  the first
alternative corresponds to the constraint $o' \leq p^nqq'$ for some $q'
\in \PP$, and the second alternative corresponds to the constraint $o' >
p^nqq'$ for some maximal position $q' \in \PP$.

For the first alternative, we start to fix the unknown $n$ by choosing $n_0 = 0$ if $p = \epsilon$,
and $n_0 = \lceil\frac{|o| \monus[\scriptstyle] |q|}{|p|}\rceil$ otherwise. 
We will show later that if $\exists o'.  p^n q < o'o$ can be satisfied by some $n$,
then it can also be satisfied using some $n\geq n_0$. For $n \geq n_0$, we
will see that
$t(C,\mu)^{n}|_{o'}$ must be 
of the form $t(C,\mu)^{n_0}|_{o''}\mu^k$ for some $o''$ and $k$.
Hence, we build the finite set of matching problems 
\[\MM_1 = \{ \sMp{t(C,\mu)^{n_0}|_{o''}}\ell 
\mid o'' \leq p^{n_0}qq' \wedge q' \in \PP \wedge p^{n_0}q < o''o\}.\]

For the second alternative where $o' >
p^nqq'$ for some maximal $q' \in \PP$, we first define the set 
$\WW = \bigcup_{k \in \nats}\VV(t|_q\mu^k)$ of variables 
that can occur below $t|_q$ when applying $\mu$ an arbitrary number
of times. Note that for substitutions with finite domains, $\WW$ is finite
and can easily be computed by iteratively applying $\mu$ on $t|_q$ until no
new variables appear. We define the second set of matching problems as
\[
\MM_2 = \{ \sMp u\ell \mid u \unlhd x\mu \wedge x \in \WW\}.
\]

We will show soundness of these matching problems by the following key lemma
which handles both alternatives.

\begin{lemma}[Connection of \rSub{fb} and $\MM_1 \cup \MM_2$]
\label{a lemma}  Let $t$ be a term, $q \in \Pos(t)$, and let $(C,\mu)$ be a context-substitution
 such that $C|_p = \BBox$ and such that $t|_q$ is not a variable.
\begin{enumerate}
\item If \rSub{fb} is satisfied with $o' \leq p^n q q'$ for some $q' \in \PP$, then a problem in $\MM_1$ is solvable.
\item If \rSub{fb} is satisfied with $o' > p^n q q'$ for some maximal $q' \in \PP$, then a problem in $\MM_2$ is solvable.
\item If a problem in $\MM_1 \cup \MM_2$ is solvable then \rSub{fb} is satisfied. 
\end{enumerate}
\end{lemma}

\begin{proof}
\begin{enumerate}
\item 
Assume
\rSub{fb} holds and let $n$, $o'$, $q' \in \PP$, and $\sigma$ be such
that $t(C,\mu)^n|_{o'} = \ell\sigma$, $o' \leq p^nqq'$, and $p^nq < o'o$.
If $p = \epsilon$ then $n_0 = 0$, and we define $o'' = o'$ and $k = n$. Hence,
using \rLE{properties}
\[
t(C,\mu)^{n_0}|_{o''}\mu^k 
= t|_{o''}\mu^k 
= t|_{o'}\mu^n
= t\mu^n|_{o'}
= t(C,\mu)^n|_{p^n}|_{o'}
= t(C,\mu)^n|_{\epsilon^n}|_{o'}
= t(C,\mu)^n|_{o'}
= \ell\sigma
\]
shows that the matching problem $\sMp{t(C,\mu)^{n_0}|_{o''}}\ell$ is solvable,
and since $o'' = o' \leq p^nqq' = p^{n_0}qq'$ and $p^{n_0}q = 
\epsilon^{n_0}q = \epsilon^nq = 
p^nq < o'o = o''o$ we also know that this matching 
problem is contained in $\MM_1$. Otherwise, $p \neq \epsilon$ and 
$n_0 = \lceil\frac{|o| \monus[\scriptstyle] |q|}{|p|}\rceil$.
W.l.o.g.~one can assume that $n \geq n_0$.\footnote{If $n < n_0$ then one can
replace $n$, $o'$, and $\sigma$ by $n + n_0$, $p^{n_0}o'$, and $\sigma\mu^{n_0}$.
These new values also satisfy \rSub{fb}.}
Hence, the position $p^{n - n_0}$ is well formed. Next, we prove that 
$o' \geq p^{n - n_0}$. Note that $o'$ cannot be parallel to $p^{n - n_0}$ as
$o' \leq p^nqq'$. If we had $o' < p^{n - n_0}$, then 
$|p^{n - n_0}| + |p^{n_0}q| = |p^nq| < |o'o| = |o'| + |o| < |p^{n - n_0}| + |o|$
shows that $n_0 \cdot |p| + |q| < |o|$, and hence yields the contradiction
$n_0 \cdot |p| 
= \lceil\frac{|o| \monus[\scriptstyle] |q|}{|p|}\rceil \cdot |p| 
< |o| \monus |q|$.
So there is some $o''$ such that $o' = p^{n-n_0}o''$ and since $o' \leq p^nqq' = 
p^{n-n_0}p^{n_0}qq'$ we know that $o'' \leq p^{n_0}qq'$. Moreover, as
$p^{n-n_0}p^{n_0}q = p^nq < o'o = p^{n-n_0}o''o$ we also know that $p^{n_0}q < o''o$.
Thus, $o'' \leq p^{n_0}qq'$ and $p^{n_0}q < o''o$ and hence, $\sMp{t(C,\mu)^{n_0}|_{o''}}\ell \in \MM_1$.
It remains to show that this matching problem is solvable which is established
using \rLE{properties}:
\[
t(C,\mu)^{n_0}|_{o''}\mu^{n - n_0}
= t(C,\mu)^{n_0}\mu^{n - n_0}|_{o''}
= t(C,\mu)^{n_0}(C,\mu)^{n-n_0}|_{p^{n-n_0}}|_{o''}
= t(C,\mu)^n|_{o'}
= \ell\sigma.
\]

\item 
We now assume that \rSub{fb} is satisfiable where $o' > p^nqq'$
for some maximal position $q' \in \PP$, and show that there is also some matching problem 
in $\MM_2$ that is solvable. 
So, let $n$, $o'$, $q'$, and $\sigma$ be such
that $t(C,\mu)^n|_{o'} = \ell\sigma$, $o' > p^nqq'$, $p^nq < o'o$, and $q'$
is a maximal position in $\PP$. Hence, $o' = p^nqq'o''$ for some $o'' \neq \epsilon$
and thus by \rLE{properties}, 
\[
t(C,\mu)^n|_{o'}
= t(C,\mu)^n|_{p^n}|_{qq'o''}
= t\mu^n|_{qq'o''}
= t|_{qq'}\mu^n|_{o''}.
\]
Since $q'$ was maximal and $o'' \neq \epsilon$ we know that $t|_{qq'}$ must be
a variable. Then one can show as in the 
proof of \cite[Thm.\ 10]{RTA08} that 
$t|_{qq'}\mu^n|_{o''} = u\mu^{k}$ for some $u \unlhd x\mu$, $x \in \WW$, and $k$.
Hence, $\sMp u\ell$ is a matching problem of $\MM_2$ and it is solvable since
\[
\ell\sigma = t(C,\mu)^n|_{o'}
= t|_{qq'}\mu^n|_{o''}
= u\mu^{k}.
\]

\item 
Assume that a problem in $\MM_1$ is solvable.
Hence, there exist $k$, $\sigma$, $o''$, and $q' \in \PP$ such that
$t(C,\mu)^{n_0}|_{o''}\mu^k = \ell\sigma$, $o'' \leq p^{n_0}qq'$,
and $p^{n_0}q < o''o$. 
Then we define $n = n_0 + k$ and $o' = p^ko''$ and achieve
\[
t(C,\mu)^n|_{o'} 
= t(C,\mu)^{n_0}(C,\mu)^k|_{p^k}|_{o''}  
= t(C,\mu)^{n_0}\mu^k|_{o''}  
= t(C,\mu)^{n_0}|_{o''}\mu^k
= \ell\sigma
\]
and moreover $p^nq = p^kp^{n_0}q < p^ko''o = o'o$. Hence, if one of the matching problems
in $\MM_1$ is solvable, then also \rSub{fb} holds.
 
We now assume that a matching problems in $\MM_2$ is 
solvable and show that then \rSub{fb} is satisfied. Here, we need the additional assumption
that $t|_q$ is not a variable. This assumption is not severe as we are interested in terms
$t$ where $t \to_q t'$, which implies that $t|_q$ is not a variable for well-formed 
TRSs.\footnote{It is also possible to define $\MM_2$ in a way that $t|_q$ can be a variable.
However, then the definitions would become even more technical. Essentially, one just would
have to perform some additional book-keeping to check whether one is strictly below $t|_q$.}
So, let $u$, $x$, $k$, $k'$, and $\sigma$ be given such that 
$x \in \VV(t|_q\mu^{k'})$, $u \unlhd x\mu$, and $u\mu^k = \ell\sigma$.
Let $o''$ and $o'''$ be positions such that $t|_q\mu^{k'}|_{o''} = x$ and $x\mu|_{o'''} = u$.
We define $n = k + k' + 1$ and $o' = p^nqo''o'''$ and show for these values that \rSub{fb}
is satisfied (again, using \rLE{properties}):
\[
t(C,\mu)^n|_{o'} 
= t(C,\mu)^n|_{p^n}|_{qo''o'''}
= t\mu^n|_{qo''o'''}
= t|_q\mu^{k'+1+k}|_{o''o'''}
= x\mu^{1+k}|_{o'''}
= u\mu^{k}
= \ell\sigma
\]
and $p^nq < p^nqo''o'''o = o'o$ since $o'' \neq \epsilon$. That $o''$ is
indeed non-empty
follows from the fact that $t|_q$ 
and thus also $t\mu^{k'}|_q$ is not a variable,
but $t\mu^{k'}|_{qo''} = t|_q\mu^{k'}|_{o''} = x$. 
\end{enumerate}
\end{proof}

Using \rLE{a lemma} it is now easy to derive the following theorem.

\begin{theorem}[Soundness of $(\ell,o,a)$ problems]
\label{a MPs correct}
Let $t \to_q t'$ and let $(C,\mu)$ be a context-substitution
such that $C|_p = \BBox$ and such that $t|_q$ is not a variable.
All reductions $t(C,\mu)^n \to_{p^nq} t'(C,\mu)^n$ are allowed
w.r.t.\ the pattern $(\ell,o,a)$ iff none of the matching 
problems in $\MM_1 \cup \MM_2$ is solvable.
\end{theorem}
Note that when encoding innermost rewriting by using forbidden patterns, the resulting
matching problems one obtains in \cite{RTA08} are essentially $\MM_1 \cup
\MM_2$.

\subsection{Deciding Loops for Forbidden Patterns of Type $(\cdot,\cdot,b)$}
\label{subsection b}
Finally, for patterns $(\ell,o,b)$, we replace the condition $p^nq =
o'o$ by $p^nq > o'o$, i.e., $o'o$ has to be strictly above the redex.
First note that $o'o \in \Pos(C^n[t])$.  Now, we consider the following
two cases: either $o'o$ ends in $t$ (i.e., $o'o \geq p^n$), or otherwise it ends in some occurrence
of $C$ (i.e., $o'o < p^n$).

In the first case there are only finitely many positions in $t$ above $q$
in which $o'o$ could end. Thus, we reduce this case to finitely many
$(\cdot,\cdot,h)$ cases. For each $\bar{q}$ above $q$ in $t$, we
consider the pattern $(\ell,o,h)$ for a reduction at position $\bar q$. Hence,
we define
\[
\MM_3 = \bigcup_{\bar{q} < q} \MM_{\bar{q}}, \text{ 
where $\MM_{\bar{q}}$ is the set of $(\ell,o,h)$ matching problems for $t$,
$\bar{q}$, and $(C,\mu)$.}
\]

In the second case  $o'o$ is a non-hole position of $C^n$, i.e.,
$p^n > o'o$.
Then $p \neq \epsilon$, since otherwise we would obtain the contradiction
$\epsilon = p^n > o'o$.
So there is a $k < n$ and a $p'''\leq p$ with $o' = p^kp'''$. Let $p''$ be the
position with $p = p'''p''$. Then we have $o < p'' p^{n_0}$ for some $n_0$.
To examine all possible choices for $o'$, we consider all prefixes $p'''$ of
$p$, i.e., all contexts $D$ with $\BBox \lhd D \unlhd C$ where
$C|_{p'''} = D$, $D|_{p''} = \BBox$,
and $p = p'''p''$. 
Let $n_0$ be the smallest
number such that $|p''| + |p^{n_0}| > |o|$ (since $p > \epsilon$, such a
number always exists).  Then we have to check whether $o < p''p^{n_0}$.
If that is not the case, then we do not result in any additional matching problems.
Otherwise, we obtain an \emph{extended matching problem}
$\eMp{D}{\ell}{C\mu}{t(C,\mu)^{n_0}\mu}$ for each $\BBox \lhd D \unlhd C$.
\begin{align*}
\MM_4 & = \{\eMp{D}{\ell}{C\mu}{t(C,\mu)^{n_0}\mu} \mid \\
 &  \phantom{ = \{} \BBox \lhd D \unlhd C, D|_{p''} = \BBox, \text{$n_0$ is least number 
  with $|p''| + n_0|p| > |o|$}, p''p^{n_0} > o\}
\end{align*}

These are the same kind of extended matching problem as for deciding
outermost loops.

\begin{definition}[Extended matching problems \cite{RTA09}]
We call a quadruple $\eMp D\ell C t$ an \emph{extended matching problem}.
It is \emph{solvable} iff there are $m$, $k$, $\sigma$, such that
$D[t(C,\mu)^m]\mu^k = \ell\sigma$.
\end{definition}

\begin{lemma}[Connection of \rSub{fb} and $\MM_3 \cup \MM_4$]
\label{b lemma}  Let $t \to_q t'$ and let $(C,\mu)$ be a context-substitution
 such that $C|_p = \BBox$.
\begin{enumerate}
\item \rSub{fb} is satisfied with $o'o \geq p^n$ iff a problem in $\MM_3$ is solvable.
\item \rSub{fb} is satisfied with $o'o < p^n$ iff a problem in $\MM_4$ is solvable.
\end{enumerate}
\end{lemma}

\begin{proof}
\begin{enumerate}
\item Suppose that a $(\ell,o,h)$ matching problem in $\MM_3$ for $\bar q < q$ is
solvable. By \rTH{h MPs correct} we obtain $\bar{q}$, $m$, $o'$, and $\sigma$
with $t(C,\mu)^n|_{o'} = \ell\sigma$ and $p^n\bar{q} = o'o$. Since
$\bar{q} < q$, this implies $p^nq > o'o$ and $o'o \geq p^n$. Thus we satisfy the case of
\rSub{fb} where $\lambda = b$ and $o'o \geq p^n$. 

Conversely, assume that there are $n$, $o'$, and $\sigma$ such that
$t(C,\mu)^n|_{o'} = \ell\sigma$, $p^nq > o'o$, and $o'o \geq p^n$. Thus,
there is some $o'' \neq \epsilon$ with $p^nq = o'oo''$. Since we are in
the case where $o'o \geq p^n$, this implies that $o''$ is a suffix
of $q$. Hence, there is some position $\bar{q}$ such that $q = \bar{q}o''$ and
$p^n\bar{q} = o'o$. As $o'' \neq \epsilon$ we know that $\bar{q} < q$ and
hence, one of the $(\ell,o,h)$ matching problems in $\MM_3$
is solvable using \rTH{h MPs correct}.

\item
Suppose that an extended matching problem in $\MM_4$ is solvable. Thus there
are $m$, $k$, and $\sigma$ such that $D[t(C,\mu)^{n_0}\mu(C\mu,\mu)^m]\mu^k =
\ell\sigma$ and $p''p^{n_0} > o$. Let $o' = p^kp'''$ and $n = k + n_0 + m + 1$. Hence, by 
\rLE{properties}
\begin{align*}
t(C,\mu)^n|_{o'}
  &= t(C,\mu)^{k + n_0 + m + 1}|_{p^kp'''}
  = t(C,\mu)^{n_0 + m + 1}\mu^k|_{p'''}
  = C[t(C,\mu)^{n_0 + m}\mu]\mu^k|_{p'''} \\
  &= D[t(C,\mu)^{n_0 + m}\mu]\mu^k  
  = D[t(C,\mu)^{n_0}\mu(C\mu,\mu)^m]\mu^k
  = \ell\sigma
\end{align*}
and moreover $p^n = p^kp^{n_0}p^mp \geq p^kpp^{n_0} = p^kp'''p''p^{n_0} >
p^kp'''o = o'o$ and thus, also $p^nq > o'o$.

In order to prove the other direction, assume that
there are $n$, $o'$, and $\sigma$ such that $t(C,\mu)^n|_{o'} =
\ell\sigma$ and $p^n > o'o$. Let $k = \lfloor\frac{|o'|}{|p|} \rfloor$.
Hence, there is some $p''' < p$ such that $o' = p^kp'''$. Since $p''' < p$,
there is also some $p''$ with $p = p'''p''$. From the fact that $o'$ is a
strict prefix of $p^n$, we obtain some $m \in \mathbb{N}$ such that $p^n =
p^kp'''p''p^m = o'p''p^m$. Thus, $o'p''p^m = p^n > o'o$ which implies $p''p^m
> o$ and so, $|p''| + |p^m| > |o|$. Hence, $m$ is greater than or equal to
the smallest number $n_0$ satisfying $|p''| + |p^{n_0}| > |o|$ and thus
$m = n_0 + m'$ for some $m' \in \nats$.  From $p^n = p^kp'''p''p^m$, we
also obtain $n = k + m + 1$. Let $D = C|_{p'''}$.
\begin{align*}
\ell\sigma
  &= t(C,\mu)^n|_{o'}
  = t(C,\mu)^{k + m + 1}|_{p^kp'''}
  = t(C,\mu)^{m + 1}\mu^k|_{p'''}
  = C[t(C,\mu)^m\mu]\mu^k|_{p'''} 
  = D[t(C,\mu)^m\mu]\mu^k \\
  &= D[t(C,\mu)^{n_0 + m'}\mu]\mu^k
  = D[t(C,\mu)^{n_0}(C,\mu)^{m'}\mu]\mu^k
  = D[t(C,\mu)^{n_0}\mu(C\mu,\mu)^{m'}]\mu^k
\end{align*}
By $m'$, $k$, $\sigma$, we obtain a solution of the extended
matching problem $\eMp{D}{\ell}{C\mu}{t(C,\mu)^{n_0}\mu}$. Note that $\BBox
\lhd D$ since otherwise $p''' = p$ which contradicts $p''' < p$. Moreover, since
$p''p^m > o$ and $|p''| + |p^{n_0}| > |o|$, we have $p''p^{n_0} > o$.
Hence, the matching problem
$\eMp{D}{\ell}{C\mu}{t(C,\mu)^{n_0}\mu}$ is contained in $\MM_4$.
\end{enumerate}
\end{proof}

\noindent
Using \rLE{b lemma}, we have proved the following theorem.

\begin{theorem}[Soundness of $(\ell,o,b)$ problems]
\label{b MPs correct}
Let $t \to_q t'$ and let $(C,\mu)$ be a context-substitution
such that $C|_p = \BBox$. All reductions
$t(C,\mu)^n\to_{p^nq} t'(C,\mu)^n$ are allowed w.r.t.\ the pattern
$(\ell,o,b)$ iff none of the matching problems in $\MM_3 \cup \MM_4$ 
is solvable.
\end{theorem}
Note that as in the innermost case, when encoding outermost rewriting by
using forbidden patterns, the resulting matching problems one obtains
in \cite{RTA09} are $\MM_3 \cup \MM_4$. So \rTH{b MPs correct} is 
a generalization of the result in \cite{RTA09}.

By combining \rCO{loops intersection} with \rTH{h MPs correct}, \rTH{a MPs
correct}, and \rTH{b MPs correct}, we 
finally obtain the following 
corollary.

\begin{corollary}[Forbidden loops are decidable]
Let $t_1 \to_{q_1} t_2 \to_{q_2} \cdots \to_{q_m} t_{m+1} = t_1(C,\mu)$ be
a loop with $C|_p = \BBox$ and let $\Pi$ be a set of forbidden patterns.
Then it is decidable whether the loop is a loop under the strategy $\Pi$.
\end{corollary}

\section{Conclusion}
\label{s:concl}

In this paper, we developed  approaches to disprove termination of
rewriting under strategies like leftmost-innermost, leftmost-outermost,
(max-)parallel-innermost, (max-)parallel-outermost, and forbidden patterns
automatically.
To this end, we introduced decision procedures which check whether a given
loop is also a loop under the respective strategy. By combining these
procedures with techniques to detect loops automatically, one obtains methods
to prove non-termination of term rewriting under these strategies.

The general idea of our decision procedures is to generate a set of
\emph{(extended) matching problems} from every loop such that one of these
matching problems is solvable iff the given loop violates the strategy. We
presented a decision problem for solvability of matching problems in
\cite{RTA08} (for extended matching problems this was done in \cite{RTA09}). 

We started with defining \emph{leftmost} matching problems in \rSC{s:leftmost loops} which
shows that it is decidable whether a loop is a leftmost loop. By combining
this result with the
decision procedures for innermost and outermost loops from \cite{RTA08,RTA09},
it is also decidable whether a loop is a
leftmost-innermost or leftmost-outermost loop.

In \rSC{s:maxpar loops} we considered parallel- and max-parallel-rewriting, where in the
latter case, \emph{all} redexes at parallel positions must be reduced simultaneously.
Similar to leftmost matching problems, here we defined
\emph{max-parallel} matching problems and showed that it is decidable whether
a given loop is also a max-parallel, a (max-)parallel-innermost, or a 
(max-)parallel-outermost loop.

Finally, in \rSC{s:forbidden loops} we extended our approach to strategies defined by
\emph{forbidden patterns} \cite{GS09}. Forbidden patterns are very expressive and
in particular, they can also be used to describe strategies such as innermost,
outermost, or context-sensitive rewriting. There are three variants of such
patterns which restrict rewriting on, above, or below certain positions of
certain subterms. For each of these classes of forbidden patterns, we showed how
to generate corresponding matching problems such that one of these matching problems is
solvable iff the given loop violates the restriction described by the pattern.
Thus, it is decidable whether a loop is also a loop under a strategy
expressed by a set of forbidden patterns.

Our results constitute the first automatic approach for disproving termination
under these strategies. Future work will be concerned with extending and
adapting our results such that they can be integrated in rewriting-based
approaches for termination analysis of programming languages (e.g.,
\cite{TOCL,TOPLAS,RTA10}).

\vspace*{.2cm}

\noindent
\textbf{Acknowledgments.}
We thank the referees for many helpful suggestions.

\vspace*{-.35cm}

\bibliographystyle{eptcs}
\bibliography{referencesIWS}

\end{document}